\documentclass[acmsmall,nonacm]{acmart}
\usepackage{thm-restate}
\usepackage[noend,ruled,linesnumbered]{algorithm2e}

\AtBeginDocument{%
  }

\setcopyright{acmcopyright}
\copyrightyear{2024}
\acmYear{2024}
\acmDOI{XXXXXXX.XXXXXXX}

\acmPrice{15.00}
\acmISBN{978-1-4503-XXXX-X/18/06}

\newcommand{\sleeping}{\textsf{SLEEPING}}
\newcommand{\local}{\textsf{LOCAL}}

\newcommand{\ID}{\textrm{ID}}
\newcommand{\DLC}{D1LC}
\newcommand{\E}{\mathbb{E}}
\DeclareMathOperator*\poly{poly}

\begin{document}

\title{Distributed Coloring in the SLEEPING Model}

\author{Fabien Dufoulon}
\affiliation{
  \institution{Lancaster University}
  \country{United Kingdom}
    }

\author{Pierre Fraigniaud}
\affiliation{
  \institution{CNRS and Université Paris Cité}
  \country{France}
    }

\author{Mikaël Rabie}
\affiliation{
  \institution{CNRS and Université Paris Cité}
  \country{France}
}

\author{Hening Zheng}
\affiliation{
  \institution{???}
  \city{Amsterdam}
  \country{Netherlands}
}


\begin{abstract}
In distributed network computing, a variant of the \local\/ model has been recently introduced, referred to as the \sleeping\/ model. In this model, nodes have the ability to decide on which round they are \emph{awake}, and on which round they are \emph{sleeping}. Two (adjacent) nodes can exchange messages in a round only if both of them are awake in that round. The \sleeping\/ model captures the ability of nodes to save energy when they are sleeping. In this framework, a major question is the following: is it possible to design algorithms that are energy efficient, i.e., where each node is awake for a few number of rounds only, without losing too much on the time efficiency, i.e., on the total number of rounds? This paper answers positively to this question, for one of the most fundamental problems in distributed network computing, namely $(\Delta+1)$-coloring networks of maximum degree~$\Delta$. We provide a randomized algorithm with average awake-complexity constant, maximum awake-complexity $O(\log\log n)$ in $n$-node networks, and round-complexity $\poly\!\log n$. 
\end{abstract}




\maketitle

\section{Introduction}

The \sleeping\/ model is a recent model of distributed computing introduced in~\cite{chatterjee2020}. This model captures the ability of computing entities to turn on and off at will for a prescribed amount of time, e.g., for saving energy. Specifically, we consider the \local\/ version of the \sleeping\/ model. That is, we are considering the standard computing framework (see, e.g., \cite{Peleg2000}) where $n\geq 1$ fault-free nodes are connected as an $n$-vertex graph $G=(V,E)$. Each node~$v\in V$ is provided with a distinct identifier, denoted by $\ID(v)$, which is the only input (and initial knowledge) given to each node before computation starts. Computation performs synchronously as a sequence of rounds. All nodes execute the same algorithm, and start at the same time, at round~1. At each round, every node sends a message to each of its neighbors, receives the messages sent by its neighbors, and performs some individual computation. The \sleeping\/ model assumes that, in addition, every node can either be \emph{awake} or \emph{sleeping}. Initially, all nodes are awake. When a node is awake, it performs as in the \local\/ model, with an additional functionality providing that node with the ability to become sleeping for a prescribed number~$r\geq 1$ of rounds. This number~$r$ is determined by the algorithm, and may depend on the state of the node when it decides to become asleep. When a node is sleeping, it cannot send messages, and all the messages sent to that node are lost. If a node turned asleep at round~$t$, for $r$ rounds, then it sleeps during rounds $t+1,\dots,t+r$, and wakes up at round $t+r+1$. Every node thus alternates between awake and sleeping periods, until it terminates. The \emph{round complexity} of an algorithm running in~$G$ is the number of rounds until all nodes of~$G$ terminates. The  \emph{awake complexity} of the algorithm is either the \emph{maximum}, taken over all nodes~$v$, of the number of rounds during which $v$ is awake before it terminates (worst case awake complexity), or the \emph{average} of the number of rounds during which the nodes of~$G$ are awake before they terminate (average awake complexity). 

\subsection{Our Results}

We consider the classical $(\Delta+1)$-coloring problem, that is, every node~$v$ of the input graph $G=(V,E)$ must output a color $c(v)\in\{1,\dots,\Delta+1\}$, where $\Delta$ denotes the maximum degree of~$G$, such that $c$ is a proper coloring of~$G$. We also consider stronger variants of the problem, including $(\deg+1)$-coloring in which every node~$v$ must output a color $c(v)\in\{1,\dots,\deg(v)+1\}$, and the list-coloring variants of these two problems: every node~$v$ is given a list $L_v$ of colors as input, of size at least $\Delta+1$, or at least $\deg(v)+1$ depending on the problem at hand, and must output a color $c(v)\in L_v$. In this paper, we solve the most constrained problem. 

\begin{theorem}
\label{thm:sleepingColoring}
There exists a randomized (Las Vegas) algorithm for $(\deg+1)$-list-coloring whose performances in the  \sleeping{} \local{} model are,  
\begin{itemize}
    \item $O(1)$ expected average round-complexity, and thus $O(1)$ expected average awake complexity;
    \item $O(\log \log n)$ worst-case awake complexity, with high probability; 
    \item $O(\poly\!\log n)$ round complexity, with high probability.
\end{itemize} 
\end{theorem}

\subsection{Related Work}

The $(\Delta+1)$-coloring problem has a long history, starting from the seminal paper by Linial~\cite{linial1992}, but we briefly survey the most recent achievements. The current state of the art for \emph{deterministic} algorithms in the \local\/ model is contrasted, as the fastest algorithm varies depending on the value of $\Delta\in\{1,\dots,n-1\}$. There are essentially three algorithms with respective round-complexities 
$\tilde{O}(\sqrt{\Delta})+O(\log^*n)$~\cite{barenboim2018,fraigniaud2016},
$O(\log n\cdot \log^2\Delta)$~\cite{ghaffari2022}, and 
$\tilde O(\log^2 n)$~\cite{GhaffariG2023}. 
As far as randomized (Las Vegas) algorithms are concerned, the best known algorithms perform, with high probability, in $\tilde{O}(\log^2 \log n)$ rounds~\cite{ChangLP18,GhaffariG2023,HalldorssonKNT22}, also for $(\deg+1)$-list-coloring.

The \sleeping\/ model was introduced only recently by Chatterjee et al.~\cite{chatterjee2020}. Despite the novelty of the model, there have already been significant efforts to understand the awake complexities of fundamental distributed problems. These include computing a maximal independent set (MIS)~\cite{chatterjee2020,ghaffari2022-1, HPR22, dufoulon2023,GP23}, a coloring~\cite{Barenboim2021}, an approximate matching or vertex cover~\cite{ghaffari2022-1}, a spanning tree~\cite{Barenboim2021} or a minimum spanning tree~\cite{AMP24}. We briefly highlight some results on the first two problems, as these are the most relevant to this work. On the one hand, \cite{Barenboim2021} gave a deterministic distributed algorithm solving $(\Delta+1)$-coloring in $O(\log \Delta+ \log^* n)$  worst-case awake complexity, which we use as a subroutine in our work. On the other hand, \cite{dufoulon2023} showed that 
maximal independent set can be solved (with a randomized algorithm) in $O(\log \log n)$ worst-case awake complexity. This was a somewhat surprising result, because it meant that the worst-case awake complexity could be exponentially better than the round complexity for maximal independent set, due to the $\Omega(\min\{\log \Delta/(\log \log \Delta), \sqrt{\log n/(\log \log n)}\})$ (worst-case) round complexity lower bound of~\cite{Kuhn_2016}. In this work, we show that $(\Delta+1)$-coloring can also be solved (with a randomized algorithm) in $O(\log \log n)$ worst-case awake complexity.

\section{Distributed Coloring}

This section is entirely dedicated to the proof of Theorem \ref{thm:sleepingColoring}. Our algorithm for $(\deg+1)$-list-coloring is split into three phases. 
\begin{itemize}
    \item In the first phase, we run a simple and elegant randomized distributed procedure directly adapted from Algorithm~19 in~\cite{barenboim2013}, for $O(\log\log n)$ iterations (and thus for $O(\log \log n)$ rounds).  Nodes that are properly colored in Algorithm \ref{alg:3} terminate, while the nodes that remain uncolored throughout the whole execution of Algorithm \ref{alg:3} continue to the second phase. We shall show that, in expectation, there are only $O(n/\poly\!\log n)$ such uncolored nodes. 
    
    \item In the second phase, we run the $O(\log^* n)$-round algorithm of~\cite{HalldorssonKNT22} to extend the coloring obtained in the first phase, and to reduce the maximum degree of the remaining uncolored nodes to $O(\log^7 n)$. 
    
    \item Finally, the third phase consists of running on the remaining uncolored nodes the state-of-the-art deterministic distributed algorithm for $(\deg+1)$-list-coloring derived from the scheme described in~\cite{Barenboim2021}. This algorithm has $O(\log\Delta + \log^*n)$ worst-case awake complexity. Therefore, since the remaining graph has maximum degree $O(\log^7 n)$, the third phase will run in $O(\log\log n)$ worst-case awake rounds.
\end{itemize}

\begin{algorithm}[tb]
\SetNoFillComment
\DontPrintSemicolon
\caption{Randomized distributed procedure for $(\deg+1)$-list-coloring directly adapted from Algorithm~19 in~\cite{barenboim2013}. Code of node $v$ parameterized by the number~$K$ of iterations, and the input list $L_v\subset \mathbb{N}\smallsetminus\{0\}$ of colors  given to~$v$.}\label{alg:3}
\For{$k=1$ {\rm \textbf{to}} $K$}{
    $c_v \gets \begin{cases}
        0\quad&\text{with probability } \frac{1}{2}\\
        \text{a random color in } L_v  \quad&\text{each with probability } \frac{1}{2\, |L_v|}
    \end{cases}$\\
    \textbf{send} color $c_v$ \textbf{to} all neighbors \textsf{// First round //} \\ 
    \textbf{receive} colors \textbf{from} neighbors  \\
    $T_v\gets \varnothing$\\
    \For{{\rm \textbf{each}}  color $c_u$ received from a neighbor $u$}{
        $T_v \gets T_v \cup \{c_u\}$\;
    }
    
    \eIf{$(c_v \neq 0) \land (c_v \not\in T_v)$}{
        adopt $c_v$ as the final color of $v$ \\
        \textbf{send}  ``\texttt{ADOPT} $c_v$'' \textbf{to} all neighbors \textsf{// Second round //} \\
        \textbf{terminate} 
    }{
        \textbf{receive} adopt messages \textbf{from} neighbors \textsf{// Second round //} \\
        \For{{\rm \textbf{each}}  message ``\texttt{ADOPT} $c_u$'' received from a neighbor $u$}{
            $L_v \gets L_v \smallsetminus \{c_u\}$\;
        }
    }
}
\end{algorithm}

More specifically, we first run Algorithm~\ref{alg:3} for $K=O(\log\log n)$ rounds. So the first phase takes merely $O(\log\log n)$ rounds. For upper bounding the round complexity and awake complexity of the second and third phases, we use the following two results, taken from \cite{HalldorssonKNT22} and \cite{Barenboim2021}, respectively.

\begin{lemma}[Theorem~1 in~\cite{HalldorssonKNT22}] \label{lemma:1}
Let $G = (V,E)$ be an $n$-node graph with maximum degree at most $\Delta$, and let $V_H$ be the nodes of G of degree at least $\log^7 \Delta$. For every positive constant $c > 0$, there exists a randomized distributed algorithm running in $O(\log^* n)$ rounds in the \local{} model satisfying that, for every $(\deg+1)$-list-coloring instance for $G$, computes a partial proper coloring of the nodes in $G[V_H]$ such that, for every node $v \in V_H$, the probability that $v$ is not colored at the end of the algorithm is at most $\Delta^{-c}$.
\end{lemma}


\begin{lemma}[Corollary of Theorem 4.2 in~\cite{Barenboim2021}] \label{lemma:2}
    Let $G = (V,E)$ be an $n$-node graph with maximum degree at most $\Delta$, and whose nodes have identifiers of size at most $N$ bits. $(\deg+1)$-list-coloring can be solved deterministically with $O(\log^*N + \log\Delta)$ (worst-case) awake-complexity, and $O(\log^*N + \Delta^2)$ round-complexity.
\end{lemma}

By setting $\Delta = n-1$ in Lemma \ref{lemma:1}, the second phase takes $O(\log^* n)$ rounds, resulting in a partial $(\deg+1)$-list-coloring that ensures that, with high probability, the remaining uncolored nodes have degree at most $O(\log^7 n)$. Subsequently, we apply Lemma \ref{lemma:2} to establish that all of the uncolored nodes of the second phase are appropriately colored during the third phase. Moreover, the third phase has $O(\log^*n + \log \log n) = O(\log \log n)$ (worst-case) awake-complexity, and $O(\log^*n + (\log^7 n)^2)$  round-complexity.

It remains to prove that the proposed algorithm has (expected) $O(1)$ node-averaged round-complexity. We first show the following. 

\begin{lemma}\label{lem:basic}
    In each iteration of the first phase, every node $v$ that has not yet terminated gets colored (and thus terminates) with probability at least $\frac{1}{4}$.
\end{lemma}

\begin{proof}
From the algorithm's description, we have 
    \begin{align*}
        \Pr[v \text{ terminates}]
        &= \Pr[(c_v \neq 0) \wedge (\forall\, u \in N(v): c_u \neq c_v)] \\
        &= \Pr[\forall\, u \in N(v): c_u \neq c_v \mid c_v \neq 0]\cdot\Pr[c_v \neq 0]\\
        &= \frac{1}{2}\cdot\Pr[\forall\, u \in N(v): c_u \neq c_v \mid c_v \neq 0].
    \end{align*}
    Moreover,
    \begin{align*}
        \Pr[c_u = c_v \mid c_v \neq 0]
        &= \Pr[c_u = c_v \mid c_v \neq 0 \wedge c_u \neq 0]\cdot\Pr[c_u \neq 0]\\
        &= \frac{1}{2}\cdot\Pr[c_u = c_v \mid c_v \neq 0 \wedge c_u \neq 0]
        \leq \frac{1}{2}\cdot\frac{1}{|L_v|}.
    \end{align*}
    Then, by union bound, $
        \Pr[\exists\, u \in N(v) : c_u = c_v \mid c_v \neq 0] \leq (|L_v| - |T_v|)\cdot (1/ 2 |L_v|)  < 1/2$.
    Thus,
    \begin{align*}
        \Pr[\forall\, u \in N(v): c_u \neq c_v \mid c_v \neq 0] 
        > 1 - \frac{1}{2}
        > \frac{1}{2}.
    \end{align*}
    Hence, as claimed, $
        \Pr[v \text{ terminates}]
        = \frac{1}{2}\cdot\Pr[\forall\, u \in N(v): c_u \neq c_v \mid c_v \neq 0]
        > (1/2) \cdot (1/2)
         = 1/4$.
\end{proof}

Let $t_f = O(\log \log n)$ be the number of rounds of the first phase. For every $i\in\{1,\dots,t_f\}$, let $X_i$ be a random variable indicating the number of nodes that run for either $2i-1$ or $2i$ rounds. By Lemma~\ref{lem:basic}, it holds that $\E[X_i] \leq n/4^{i-1}$. Moreover, the expected node-averaged round complexity is upper bounded by 
\[
\E\left[\sum_{1 \leq i \leq t_f} (X_i / n \cdot  2i) + X_{t_f+1} / n \cdot O(\poly \log n)\right], 
\]
as we have already shown that the round complexity of our  algorithm is $O(\poly \log n)$. By linearity of expectation, the expected node-averaged round complexity is upper bounded by 
\[
\frac{1}{n} \cdot \sum_{1 \leq i \leq t_f} (\E[X_i] \cdot  2i) + \E[X_{t_f+1}] \cdot O(\poly \log n). 
\]
Choosing $t_f$ large enough ensures that the expected node-averaged round complexity is constant, as desired.
\qed

\section{Conclusion}

Previous research~\cite{dufoulon2023} on the \sleeping\/ model showed that the $O(\log \log n)$ worst-case awake complexity of maximal independent set could be exponentially smaller than its worst-case round complexity. In this work, we show a similar $O(\log \log n)$ worst-case awake complexity upper bound for $(\Delta+1)$-coloring, but this is only polynomially better than the state-of-the-art round complexity upper bound for the problem. This naturally raises the following question: does the $(\Delta+1)$-coloring problem admit a randomized algorithm with $o(\log \log n)$ worst-case awake complexity? And if not, can we hope for a general $\Omega(\log \log n)$ worst-case awake complexity lower bound for a large class of local distributed problems?

\begin{acks}
 Pierre Fraigniaud and Mikaël Rabie received additional support from ANR project DUCAT (ANR-20-CE48-0006). 
\end{acks}

\bibliographystyle{plain}

\end{document}